\documentclass[aip,apl,prl,preprint,floatfix, amsmath,amssymb,aps]{revtex4-2}

\usepackage[utf8]{inputenc}

\usepackage[dvipsnames]{xcolor}
\definecolor{color1}{RGB}{230,57,70}
\definecolor{color2}{RGB}{7,36,66}
\definecolor{color3}{RGB}{50,23,77}
\usepackage{tikz}

\usepackage{amssymb,amsthm,amsfonts,mathtools}
\usepackage{newtxtext,newtxmath,dsfont,microtype}
\usepackage{xspace}
\usepackage{tabularx,enumerate,booktabs}
\usepackage{thmtools}
\usepackage{thm-restate}

\input{braket}

\newcommand\prob\textsc
\makeatletter
\newcommand{\problemtitle}[1]{\gdef\@problemtitle{#1}}
\newcommand{\probleminput}[1]{\gdef\@probleminput{#1}}
\newcommand{\problemquestion}[1]{\gdef\@problemquestion{#1}}
\newcommand{\problempromise}[1]{\gdef\@problempromise{#1}}
\DeclareDocumentEnvironment{problem}{o o}{
\problemtitle{\IfValueT{#1}{#1}}\probleminput{}\problempromise{}\problemquestion{}
}{
  \par\addvspace{.5\baselineskip}
  \noindent
  \begin{tabularx}{\linewidth}{@{\hspace{\parindent}} l X c}
    \multicolumn{2}{@{\hspace{\parindent}}l}{\prob{\@problemtitle}} \\
    \textbf{Input:} & \@probleminput \\
\ifx\@problempromise\empty%
\else%
    \textbf{Promise:} & \@problempromise \\
\fi%
    \textbf{Question:} & \@problemquestion
  \end{tabularx}
  \par\addvspace{.5\baselineskip}
}
\makeatother

\newcommand{\defas}{\coloneqq}

\newcommand{\lham}{{\textup{\textsc{$k$-local Hamiltonian}}}\xspace}

\newcommand{\mham}{{$\mathcal{M}$\textup{\textsc{-Hamiltonian}}}\xspace}
\newcommand{\kham}{{$\mathcal{K}$\textup{\textsc{-Hamiltonian}}}\xspace}
\newcommand{\precmham}{{\textup{\textsc{Precise-}}$\mathcal{M}$\textup{\textsc{-Hamiltonian}}}\xspace}

\newcommand{\yesHam}{{\textup{\textsc{Yes-Hamiltonian}}}\xspace}
\newcommand{\succinctYesHam}{{\textup{\textsc{Succinct-Yes-Hamiltonian}}}\xspace}
\newcommand{\flagid}{{\textup{\textsc{Flag Identification}}}\xspace}

\newcommand{\QMA}{{\textup{QMA}}\xspace}
\newcommand{\QMAcs}{{\textup{QMA}$(c,s)$}\xspace}
\newcommand{\PGQMA}{{\textup{PGQMA}}\xspace}

\newcommand{\QMAEXP}{{\textup{QMA\textsubscript{EXP}}}\xspace}
\newcommand{\QMAEXPcs}{{\textup{QMA\textsubscript{EXP}}$(c,s)$}\xspace}
\newcommand{\PreciseQMA}{{\textup{PreciseQMA}}\xspace}
\newcommand{\PSPACE}{{\textup{PSPACE}}\xspace}
\DeclareMathOperator{\lmin}{\lambda_\mathrm{min}}

\newcommand{\ii}{\mathrm{i}}
\newcommand{\ee}{\mathrm{e}}

\newcommand{\1}{\mathds 1}

\newcommand{\YES}{{\textup{YES}}\xspace}
\newcommand{\NO}{{\textup{NO}}\xspace}

\usepackage{hyperref}
\usepackage[capitalise]{cleveref}
\Crefname{lemma}{Lemma}{Lemmas}
\Crefname{proposition}{Proposition}{Propositions}
\Crefname{definition}{Definition}{Definitions}
\Crefname{theorem}{Theorem}{Theorems}
\Crefname{conjecture}{Conjecture}{Conjectures}
\Crefname{corollary}{Corollary}{Corollaries}
\Crefname{example}{Example}{Examples}
\Crefname{section}{Section}{Sections}
\Crefname{appendix}{Appendix}{Appendices}
\Crefname{figure}{Fig.}{Figs.}
\Crefname{equation}{Eq.}{Eqs.}
\Crefname{table}{Table}{Tables}
\Crefname{item}{Property}{Properties}
\Crefname{remark}{Remark}{Remarks}

\newtheorem{theorem}{Theorem}[section]
\newtheorem{lemma}[theorem]{Lemma}

\newtheorem{corollary}[theorem]{Corollary}
\newtheorem{definition}[theorem]{Definition}

\newcommand{\identity}{\mathds{1}}

\newcommand{\Htarget}{H_{\mathrm{target}}}

\newcommand{\HLS}{H_{\mathrm{LS}}}
\newcommand{\Huniv}{H_{\mathrm{univ}}}
\newcommand{\Htiuniv}{H_{\mathrm{succ}}}

\newcommand{\Hsim}{H_{\mathrm{sim}}}
\newcommand{\TPE}{T_\mathrm{PE}}

\newcommand{\Heff}{H_{\mathrm{eff}}}
\newcommand{\Hmk}{H_{\mathrm{MK}}}
\newcommand{\Hin}{H_{\mathrm{in}}}

\newcommand{\Hprop}{H_{\mathrm{prop}}}
\newcommand{\Hout}{H_{\mathrm{out}}}
\newcommand{\Hclock}{H_{\mathrm{clock}}}

\newcommand{\field}{\mathds}

\DeclareMathOperator{\BigO}{O}
\DeclareMathOperator{\poly}{poly}
\DeclareMathOperator{\Herm}{Herm}
\DeclareMathOperator{\rank}{rank}

\renewcommand{\paragraph}[1]{\textbf{\textit{#1}}}

\usepackage{natbib}

\begin{document}

\title{General conditions for universality of Quantum Hamiltonians}

\author{Tamara Kohler}
\affiliation{Department of Computer Science, University College London, UK}
\author{Stephen Piddock}
\affiliation{School of Mathematics, University of Bristol, UK}
\affiliation{Heilbronn Insitute for Mathematical Research, Bristol, UK}
\author{Johannes Bausch}
\affiliation{Department of Applied Mathematics and Theoretical Physics, University of Cambridge, UK}
\author{Toby Cubitt}
\affiliation{Department of Computer Science, University College London, UK}

\date{}

\begin{abstract}
Recent work has demonstrated the existence of \emph{universal} Hamiltonians -- simple spin lattice models that can simulate any other quantum many body system to any desired level of accuracy.
Until now proofs of universality have relied on explicit constructions, tailored to each specific family of universal Hamiltonians.
In this work we go beyond this approach, and completely classify the simulation ability of quantum Hamiltonians by their complexity classes.
We do this by deriving necessary and sufficient complexity theoretic conditions characterising universal quantum Hamiltonians.
Although the result concerns the theory of analogue Hamiltonian simulation -- a promising application of near-term quantum technology -- the proof relies on abstract complexity theoretic concepts and the theory of quantum computation.
As well as providing simplified proofs of previous Hamiltonian universality results, and offering a route to new universal constructions, the results in this paper give insight into the origins of universality. For example, finally explaining the previously noted coincidences between families of universal Hamiltonian and classes of Hamiltonians appearing in complexity theory.
\end{abstract}

\maketitle

\clearpage
\tableofcontents


\section{Introduction}

Recent work has precisely defined what it means for one quantum system to simulate the full physics of another~\cite{Cubitt:2017}, and demonstrated that---within very demanding definitions of what it means for one system to simulate another---there exist families of local Hamiltonians that are \emph{universal}, in the sense that they can simulate all other Hamiltonians to any accuracy desired.
This rigorous mathematical framework of Hamiltonian simulation not only gives a theoretical foundation for describing analogue Hamiltonian simulation -- one of the most promising applications of quantum computing in the NISQ (noisy, intermediate-scale quantum) era.
It also unifies many previous Hamiltonian complexity results, and implies new ones~\cite{Cubitt:2017}.
And it has even found applications in constructing the first rigorous holographic dualities between local Hamiltonians, providing toy models of AdS/CFT duality in quantum gravity~\cite{HQECC-local}.

However, previous constructions of universal Hamiltonians have relied heavily on using perturbation gadgets, and constructing complicated `chains' of simulations to prove that simple models are indeed universal.
For example, the original paper on universal quantum Hamiltonians used a chain of more than 10 simulations and perturbation gadget constructions to prove universality of the 2D general Heisenberg and XY models.
This work was recently extended, with the first construction of a translationally invariant universal family of quantum Hamiltonians~\cite{PiddockBausch}, again using a complex, carefully engineered construction.

In~\cite{kohler2020translationallyinvariant}, the present authors developed a new method of proving universality and used it to construct the first  universal models acting in one dimension.
This new method pointed to a connection between universality and complexity.
This connection is not entirely surprising.
Indeed, a rigorous complexity theoretic characterisation of universal Hamiltonians has already been demonstrated in the classical case~\cite{Cubitt:2016}.
Essentially, \cite{Cubitt:2016} showed that if a family of classical Hamiltonians has a ground state energy problem that is NP-hard, then it is necessarily also capable of simulating the complete physics of any other classical Hamiltonian. The converse implication is immediate.

In the quantum setting, there were hints that a similar result might hold.
The classes of two-qubit interactions that are universal for simulating all, stoquastic, and classical Hamiltonians, respectively, were fully characterised in~\cite{Cubitt:2017}, and turned out to coincide precisely with the classes of interactions that have QMA-, StoqMA and NP-complete ground state energy problems.
However, the proofs of these two classifications were independent, and it was certainly possible this coincidence only applied in the case of qubits, as the proof techniques relied critically on having only two-qubit interactions.
Furthermore, the more complicated non-commutative structure of quantum Hamiltonians made it impossible to replicate the classical approach of~\cite{Cubitt:2017} to proving a relationship between complexity and simulation.

In this paper, by extending the simulation technique developed in~\cite{kohler2020translationallyinvariant}, we resolve this.
We derive necessary and sufficient complexity-theoretic conditions for a family of Hamiltonians to be an efficient universal model, relating this directly to complexity-theoretic properties of the ground state.

\section{Main Results}

Our main result is a complexity theoretic classification of which families of Hamiltonians are efficient universal models:

\begin{restatable}[Universality Classification]{thm}{classification} \label{th:classification_intro}
A family of Hamiltonians, $\mathcal{M}$, is an efficient universal model iff \mham is \QMA-complete under faithful reductions, and $\mathcal{M}$ is closed.
\end{restatable}


The first \emph{faithful reduction} condition for a family of Hamiltonians $\mathcal{M}$ to be an efficient universal model is related to the complexity of \mham -- the problem of deciding whether or not a Hamiltonian in $\mathcal{M}$ has a low energy ground state.
We prove a connection between \QMA-completeness of \mham, and universality.
However, \QMA-completeness alone isn't enough for a model to be universal.
We demonstrate that in order for $\mathcal{M}$ to be a universal model, \mham must be \QMA-complete under \emph{faithful} reductions.
We say that a reduction from a problem in \QMA to \mham is \emph{faithful} if it maps the subspace picked out by a \QMA-verification circuit to the low energy subspace of the Hamiltonian - therefore preserving some of the structure of the verification circuit and the witness. (For a rigorous definition of faithfulness see \cref{sec:faithfulness}.)

The second condition, \emph{closure}, relates to combining different Hamiltonians from the same model.
We say a model, $\mathcal{M}$, is closed if, given $H_A^{(1)},H_B^{(2)} \in \mathcal{M}$, acting on (possibly overlapping) sets of qudits $A,B$ respectively, there exists a Hamiltonian $H^{(3)} \in \mathcal{M}$ which can simulate $H_A^{(1)}+H_B^{(2)}$.

Furthermore, in this paper we provide a recipe for modifying history state Hamiltonians so that the canonical reduction from a \QMA-problem to the history state Hamiltonian is faithful.
Therefore, all that remains to show that a family of history state Hamiltonians is universal is to demonstrate closure.

We also derive two corollaries, giving complexity-theoretic conditions for families of Hamiltonians to be universal models which aren't efficient in the sense of~\cite{Cubitt:2017}, but are nonetheless interesting.
These corollaries, along with the main theorem, give a complete classification of all known universal models.

\section{Discussion}

The most obvious implication of our results is that they provide a new route for proving universality of families of Hamiltonians.
All previous universality results \cite{Cubitt:2017,PiddockBausch,kohler2020translationallyinvariant,zhou2021strongly} relied on explicit constructions, tailored to particular universal models.
These constructions showed that given some arbitrary target Hamiltonian $\Htarget$, a Hamiltonian from the universal model could be constructed which simulated $\Htarget$.
These constructions typically drew on techniques from Hamiltonian complexity theory, such as perturbation gadgets and history states, but the proofs of universality required substantial additional work.

That additional work can now in many cases be side-stepped, by using the extensive existing work classifying the complexity of the local Hamiltonian problem \cite{Cubitt2013,gottesman2009,Kempe:2003,Kitaev2002}, along with our main result (as well as our recipe for modifying history state Hamiltonians). 
The remaining step is to demonstrate closure - for some families of Hamiltonians this step is trivial (it follows from the definition of the model), for others it will require some work - however demonstrating closure will always be simpler than demonstrating universality, since closure requires the ability to simulate a very limited class of Hamiltonians.

It should be noted that this method of proving universality is not constructive, in the sense that it doesn't tell you how to simulate a given target Hamiltonian with a Hamiltonian from the a universal model.
However, we do not view this as a drawback when compared to previous methods of proving universality.
The previous methods for proving universality were theoretically constructive - in the sense that for any target Hamiltonian they provided a mathematical description of a Hamiltonian from the universal model which could simulate it.
However, in reality the cost of leaving the target Hamiltonian completely general in the previous proofs was that the simulations constructed had to be very complex - putting them out of reach of current experimental limitations.
The benefit of our work is that it gives a simple route to proving universality, so now the work on constructing explicit simulations can focus on simulations which are experimentally feasible.

The problem of constructing universal Hamiltonians where the simulations are experimentally feasible is challenging, but solving it is of fundamental importance.
Current constructions fall down in either requiring precise control of interaction strengths across many orders of magnitude, or in requiring very large local Hilbert space dimension (and in some cases both).
We believe our results provide an extra tool for tackling this challenge.
As we've already mentioned, our results give a new technique for proving universality.
This could be applied to, for example, the construction in \cite{Bausch2017} of a low-dimensional spin chain with \QMAEXP-complete ground state energy problem to check whether the model is universal. 
The construction in \cite{Bausch2017} is complex, and checking universality via previous methods would be a difficult technical challenge. With the aid of our results it becomes significantly more straightforward.
If it is shown to be universal it would be a universal model which doesn't require tuning of individual interaction strengths, and with local Hilbert state dimension orders of magnitude smaller than any previously known model. 
This would be an important step towards experimentally feasible universal models.

The other direction our results could be used in is investigating the complexity of models which are currently used as analogue simulators, to determine whether there is hope of using them to construct universal models. 
One platform that is currently used for analogue Hamiltonian simulation is Rydberg atoms \cite{Morgado_2021,Scholl_2021,Wu_2021}.
It has been suggested \cite{Morgado_2021} that this might have promise as a universal simulator as in certain regimes the platform naturally encodes the $XY$-Hamiltonian, which is known to be universal \cite{Cubitt:2017}.
However, the proof of universality can not be used to construct a universal simulator using Rydberg atoms as it involves complicated chains of perturbative simulations, requiring precise control beyond the reach of current experiments.
An alternative approach to investigating the use of Rydberg atoms as universal simulators could be to investigate the complexity of the ground state energy problem of Rydberg interactions where the control over interaction strengths is limited to what is experimentally feasible. 
As outlined above, demonstrating complexity in this regime is likely to be more straightforward than directly proving universality.
If, within the limitations on interaction strength, it is possible to demonstrate hardness of the ground state energy problem, that motivates attempts to look for simple universal constructions.
If it is not possible to demonstrate hardness of the ground state energy problem within current experimental limitations, it may be possible to determine how much experimental techniques have to advance in order to overcome the barrier and achieve universality. 

Finally, the relationship between complexity and universality is interesting from a fundamental physics standpoint. 
It was already clear that universality implied complexity - since universal models must be able to simulate \emph{all} quantum many body physics.
However the reverse direction was not obvious.
Our results show that if the problem of deciding whether a Hamiltonian in $\mathcal{M}$ has a low energy ground state is hard for a quantum computer, then $\mathcal{M}$ must be rich enough to capture all quantum many body physics. 

The remainder of the paper contains the technical proofs of our main results. 
In \cref{sec:preliminaries} we cover some necessary technical background on universal Hamiltonians and complexity theory.
The notion of a faithful reduction is outlined in detail in \cref{sec:faithfulness}; in this section we also present a family of Hamiltonians which we demonstrate is \QMA-complete under faithful reductions.
Finally our main theorem is proved in \cref{general-condition}.

\section{Technical preliminaries} \label{sec:preliminaries}
\subsection{Universal Hamiltonians}

In \cite{Cubitt:2017} a rigorous definition of what it means for one quantum system to simulate another was developed:

\begin{definition}[{Approximate simulation,~\cite[Def.~23]{Cubitt:2017}}]\label{app-sim}
Let $\Delta,\eta,\epsilon>0$.
A Hamiltonian $H'$ is a $(\Delta, \eta, \epsilon)$-simulation of the Hamiltonian $H$ if there exists a local encoding $\mathcal{E}(M)=V(M \otimes P + \overline{M} \otimes Q){V}^\dagger$ such that
\begin{enumerate}[i.]
\item \label{app-sim:1} There exists an encoding $\tilde{\mathcal{E}}(M)=\tilde{V}(M \otimes P + \overline{M} \otimes Q)\tilde{V}^\dagger$ into the subspace $S_{\tilde{\mathcal E}}$ such that $S_{\tilde{\mathcal{E}}} = S_{\leq \Delta(H')}$ and $\|\tilde{V} - V\| \leq \eta$; and
\item \label{app-sim:2} $\|H'_{\leq \Delta} - \tilde{\mathcal{E}}(H)\| \leq \epsilon$.
\end{enumerate}
\end{definition}

\noindent where an encoding is a map of the form $\mathcal{E}(A) =V \left(A \otimes P + \overline{A} \otimes Q \right) V^\dagger $ for $V$ an isometry, and a local encoding maps local observables in the target system to local observables in the simulator system, defined as:

\begin{definition}[Local subspace encoding (Definition 13 from~\cite{Cubitt:2017})]\label{def:local-encoding}
Let
\[
\mathcal{E}: \mathcal{B}\left(\otimes_{j=1}^N \mathcal{H}_j \right) \rightarrow \mathcal{B}\left(\otimes_{j=1}^{N} \mathcal{H}'_j \right)
\]
 be a subspace encoding. We say that the encoding is local if for any operator $A_j \in \Herm(\mathcal{H}_j)$ there exists $A'_j \in \Herm(\mathcal{H}'_j)$ such that
\[
\mathcal{E}(A_j \otimes \identity) = (A'_j \otimes \identity)\mathcal{E}(\identity).
\]
\end{definition}

\noindent Note that the role of $\tilde{\mathcal E}$ in \cref{app-sim} is to provide an \emph{exact} simulation as per \cite[Def.~20]{Cubitt:2017}. However, it might not always be possible to construct this encoding in a local fashion. The local encoding $\mathcal E$ in turn approximates $\tilde{\mathcal E}$, such that the subspaces mapped to by the two encodings deviate by at most $\eta$. $\epsilon$ controls how much the eigenvalues are allowed to differ.

In~\cite{Cubitt:2017} it is shown that approximate Hamiltonian simulation preserves important physical properties. We recollect the most important ones in the following.
\begin{lemma}[{\cite[Lem.~27, Prop.~28, Prop.~29]{Cubitt:2017}}] \label{physical-properties}
Let $H$ act on $(\field{C}^d)^{\otimes n}$.
  Let $H'$ act on $(\field{C}^{d'})^{\otimes m}$, such that $H'$ is a $(\Delta, \eta, \epsilon)$-simulation of $H$ with corresponding local encoding $\mathcal{E}(M) = V(M \otimes P + \overline{M} \otimes Q)V^\dagger$.
  Let $p = \rank(P)$ and $q = \rank(Q)$.
  Then the following holds true.
  \begin{enumerate}[i.]
  \item Denoting with $\lambda_i(H)$ (resp.\ $\lambda_i(H')$) the $i$\textsuperscript{th}-smallest eigenvalue of $H$ (resp.\ $H'$), then for all $1 \leq i \leq d^N$, and all $(i-1)(p+q) \leq j \leq i (p+q)$, $|\lambda_i(H) - \lambda_j(H')| \leq \epsilon$.
  \item The relative error in the partition function evaluated at $\beta$ satisfies
    \begin{equation}
      \frac{|\mathcal{Z}_{H'}(\beta) - (p+q)\mathcal{Z}_H(\beta) |}{(p+q)\mathcal{Z}_H(\beta)} \leq \frac{(d')^m \ee^{-\beta \Delta}}{(p+q)d^N \ee^{-\beta \|H\|}} + (\ee^{\epsilon \beta} - 1).
    \end{equation}
  \item For any density matrix $\rho'$ in the encoded subspace for which $\mathcal{E}(\identity)\rho' = \rho'$, we have
    \begin{equation}
      \|\ee^{-\ii H't}\rho'\ee^{\ii H't} - \ee^{-\ii \mathcal{E}(H)t}\rho'\ee^{\ii \mathcal{E}(H)t}\|_1 \leq 2\epsilon t + 4\eta.
    \end{equation}
  \end{enumerate}
\end{lemma}

If we are interested in whether an infinite family of Hamiltonians can be simulated by another, the notion of overhead becomes interesting: if the system size grows, how large is the overhead necessary for the simulation, in terms of the number of qudits, operator norm or computational resources?
We capture this notion in the following definition.
\begin{definition}[{Simulation,~\cite[Def.~23]{Cubitt:2017}}]\label{def:efficient-sim}
We say that a family $\mathcal{F}'$ of Hamiltonians can simulate a family $\mathcal{F}$ of Hamiltonians if, for any $H \in \mathcal{F}$ and any $\eta, \epsilon > 0$ and $\Delta \geq \Delta_0$ (for some $\Delta_0 > 0$), there exists $H' \in \mathcal{F}'$ such that $H'$ is a $(\Delta, \eta, \epsilon)$-simulation of $H$.

We say that the simulation is efficient if, in addition, for $H$ acting on $n$ qudits and $H'$ acting on $m$ qudits, $\|H'\| = \poly(n, 1 / \eta,1 / \epsilon,\Delta)$ and $m = \poly(n, 1 / \eta,1 / \epsilon,\Delta)$; $H'$ is efficiently computable given $H$, $\Delta$, $\eta$ and $\epsilon$; each local isometry $V_i$ in the decomposition of $V$ is itself a tensor product of isometries which map to $\BigO(1)$ qudits; and there is an efficiently constructable state $\ket{\psi}$ such that $P \ket{\psi} = \ket{\psi}$.
\end{definition}

\Cref{def:efficient-sim} naturally leads to the question in which cases a family of Hamiltonians is \emph{so} versatile that it can simulate any other Hamiltonian: in that case, we call the family \emph{universal}.
\begin{definition}[{Universal Hamiltonians~\cite[Def.~26]{Cubitt:2017}}]
We say that a family of Hamiltonians is a universal simulator---or simply is universal---if any (finite-dimensional) Hamiltonian can be simulated by a Hamiltonian from the family.
We say that the universal simulator is efficient if the simulation is efficient for all local Hamiltonians.
\end{definition}

\subsection{Complexity theory}

\begin{definition}[\QMAcs~\cite{Kitaev2002}]
A promise problem $A=A_\YES\cup A_\NO$ is in \QMAcs if and only if there exists a polynomially-bounded function $p$ and a quantum polynomial time verifier $V$ such that for all $n$ and all $x \in \{0,1\}^n$:
\begin{itemize}
\item If $x \in A_\YES$, there exists a $p(n)$-qubit quantum state $\ket{\psi}$ such that\linebreak $\textnormal{Pr}[V \text{ accepts } (x,\ket{\psi})] \geq c$
\item If $x \in A_\NO$ for all $p(n)$-qubit quantum states $\ket{\psi}$, $\textnormal{Pr}[V \text{ accepts } (x,\ket{\psi})] \leq s$
\end{itemize}
\end{definition}

\QMAcs where $c$ and $s$ are separated by an inverse polynomial is the class \QMA~\cite{Kitaev2002}.
\QMAcs where $c$ and $s$ are separated by an exponentially small gap is the class \PreciseQMA.
\PreciseQMA is known to be equal to the class \PSPACE~\cite{Fefferman2016}.

\begin{definition}[\QMAEXPcs~\cite{gottesman2009quantum}]
A promise problem $A=A_\YES\cup A_\NO$ is in \QMAEXPcs if and only if there exists a $k$ and a quantum exponential time verifier $V$ such that for all $n$ and all $x \in \{0,1\}^n$:
\begin{itemize}
\item If $x \in A_\YES$, there exists a $2^{n^k}$-qubit quantum state $\ket{\psi}$ such that\linebreak $\textnormal{Pr}[V \text{ accepts } (x,\ket{\psi})] \geq c$
\item If $x \in A_\NO$ for all $2^{n^k}$-qubit quantum states $\ket{\psi}$, $\textnormal{Pr}[V \text{ accepts } (x,\ket{\psi})] \leq s$
\end{itemize}
\end{definition}

\QMAEXPcs where $c$ and $s$ are separated by an inverse exponential is the class \QMAEXP.

The canonical problem in Hamiltonian complexity is \lham.
\begin{problem}[\lham$(g)$]
\probleminput{
Local Hamiltonian $H=\sum_{i=1}^m h_i$ on an $N$-partite Hilbert space of constant local dimension, and $m\le \poly N$.
Each $h_i \defas h_{S_i} \otimes \1_{S_i^c}$ acts non-trivially on at most $|S_i| \le k$ sites, and $\| h_i \| \le 1$.
Two numbers $\alpha,\beta>0$ with $\beta-\alpha \ge g(N)$.
}
\problempromise{
The ground state energy $\lmin(H)$ either $\ge\beta$, or $\le\alpha$.
}
\problemquestion{
\YES if  $\lmin(H) \le \alpha$, else \NO.
}
\end{problem}

\lham$(1/\poly)$ is \QMA-complete for $k \geq 2$. \lham$(1/\exp)$ is \PreciseQMA-complete.

We can also consider the special case where the set of interaction terms and / or the geometry of the interaction graph is restricted (which can implicitly constrain the family's locality $k$).
\begin{definition}[\mham]
The \lham$(1/\poly)$ problem, where the Hamiltonian is restricted to belong to $\mathcal{M}$, some (possibly infinite) family of Hamiltonians.
\end{definition}

\begin{definition}[\precmham]
The \lham$(1/\exp)$ problem, where the Hamiltonian is restricted to belong to $\mathcal{M}$, some (possibly infinite) family of Hamiltonians.
\end{definition}

%

\section{Faithfulness condition}  \label{sec:faithfulness}



\subsection{Faithful Hamiltonian reductions}
The acceptance operator, $Q(U)$, of a \QMA-verification circuit, $U$, is defined as~\cite[Th.~3.6]{marriott2005quantum}:
\begin{equation} \label{eq:accept_op}
Q(U) =  \bra{0}^{\otimes m}U^\dagger \Pi_\textrm{out} U \ket{0}^{\otimes m}
\end{equation}
where $U$ requires $m$ ancillas, initialised in the $\ket{0}$ state.

\begin{definition}[Gapped acceptance operators]\label{def:acc-op}
Consider a promise problem $A$ which can be verified by a unitary circuit $U$, with completeness probability $c$.
Let $x\in A$ with $n=|x|$ the size of the instance. We say the acceptance operator is gapped if $\lambda_x$, the largest eigenvalue of $Q(U_x)$ which is less than completeness $c$, it holds that  $c-\lambda_x > 1/\poly(n)$.
\end{definition}

In other words, \cref{def:acc-op} means that any state with acceptance probability below the completeness threshold already lies significantly below it, namely $1/\poly$ bounded away.

Note there is a subtle difference between the promise gap and the question of whether or not the acceptance operator is gapped.
For any \NO instance of a problem $A \in$ \QMA the definition of \QMA trivially implies that the acceptance operator is gapped (since the acceptance probability is below the soundness threshold $\lambda_x \le s$, and $c-s > 1/\poly$ by definition).
However, for \YES instances it is possible to have an acceptance operator which is not gapped.
We will see that for \YES instances the question of whether or not the acceptance operator is gapped is related to the \emph{spectral gap} of a Hamiltonian, rather than the promise gap.

The idea of requiring a gap in the spectrum of proof systems has arisen before in the Hamiltonian complexity literature, first in~\cite{aharonov2008pursuit} in the definition of the class \PGQMA (Polynomially Gapped \QMA).\footnote{\PGQMA is a similar class to \QMA with the added condition that the acceptance operator of the verification circuit has an inverse polynomial spectral gap. Note that for \PGQMA the gap is required to be between the lowest and second lowest eigenvalue, unlike in our definition.}
The notion of a gap in the spectrum of the proof system is again seen to be related to the \emph{spectral gap} of a Hamiltonian, as \lham with the added promise that the spectral gap of the Hamiltonian is inverse polynomial is complete for \PGQMA~\cite{aharonov2008pursuit}.

\begin{definition}[Faithful Hamiltonian reduction] \label{faithfulness_definition}
Let $A = A_\YES \cup A_\NO$ be a promise problem which can be verified by a family of circuits, $U$, of length $T$, with completeness probability $c$.
The acceptance operator $Q(U)$ is as defined in \cref{eq:accept_op}.
Consider a reduction from $A$ to the \mham problem.
For a verification circuit with gapped $Q(U)$, we say the reduction is \emph{faithful} with respect to $U$ if for all instances $x \in A_\YES$ there exists a Hamiltonian $H_x \in \mathcal{M}$ acting on $\poly(n)$ qudits (where $n=|x|$) such that for the low energy subspace, \[
\mathcal{S}_0 \defas\mathrm{span} \left\{\ket{\phi} : H_x\ket{\phi} = \tilde{\lambda} \ket{\phi}, \tilde{\lambda} \leq \frac{\kappa (1-c)}{T+1} \right\}
\]
for $\kappa =1/\poly(n)$, the following holds.
\begin{enumerate}[(i).]
\item  \label{first_faithfulness_condition} $
\| \Pi_{\mathcal{S}_0} - \Pi_{\mathcal{E}(\mathcal{L})} \| \le \eta
$
where
\begin{itemize}
\item $\eta < 1$ can be made arbitrarily small,
\item $\Pi_{\mathcal{S}}$ denotes the projector onto the subspace $\mathcal{S}$,
\item $\mathcal{E}$ is some local encoding (independent of the problem being encoded),
\item $\mathcal{L} \defas \mathrm{span} \left\{\ket{\psi} : Q(U)\ket{\psi} = {\lambda} \ket{\psi}, \lambda \geq c \right\}$.
\end{itemize}
\item \label{second_faithfulness_condition} The spectral gap above the subspace $\mathcal{S}_0$ is $\Omega(1/\poly(T))$.
\end{enumerate}
\noindent
For $x \in A_\NO$, there are no conditions on $H_x$.
Similarly, if $Q(U)$ is not gapped there are no conditions on the $H_x$.
\end{definition}

Note that the concept of a faithful reduction is a property of a particular verification circuit,\footnote{Although any equivalent model of computation could be substituted into the definition.} not of the problem itself.

\begin{definition} \label{def:qma-under-faithful}
We say that \mham is \QMA-complete under faithful reductions if for all $A \in $ \QMA and for any polynomial time \QMA-verification circuit $U$ which verifies $A$, there exists a reduction from $A$ to the \mham problem which is faithful with respect to $U$.
\end{definition}

\subsection{The modified Kitaev Hamiltonian}

The Hamiltonian we use to prove necessity of the faithfulness condition is a modification of the 5-local Hamiltonian shown to be \QMA-complete in~\cite{Kitaev2002}.
Note that this choice is convenient, but the procedure we set out here to demonstrate faithfulness could be applied to any history-state Hamiltonian in the literature.

The original 5-local Hamiltonian is a ``circuit-to-Hamiltonian'' mapping, given by
\begin{equation}
H_{\mathrm{K}} = \Hin + \Hprop +  \Hout + \Hclock
\end{equation}
where the Hamiltonian is acting on the Hilbert space
\begin{equation}
\mathcal{H} \coloneqq \mathcal{H}_Q \otimes \mathcal{H}_C = (\mathbb{C}^2)^{\otimes n} \otimes (\mathbb{C}^2)^{\otimes T+1}=  \mathbb{C}^2 \otimes (\mathbb{C}^2)^{\otimes |W|}\otimes (\mathbb{C}^2)^{\otimes |A|} \otimes (\mathbb{C}^2)^{\otimes T}
\end{equation}
and
\begin{align}
\Hin &= \Pi_1^{(1)}\otimes \ketbra{0}_1^c + \sum_{j \in A} \Pi_j^{(1)}\otimes \ketbra{0}_1^c\\
\Hout &= \Pi_1^{(0)}\otimes \ketbra{1}_{T}^c\\
\Hclock &= \identity \otimes \sum_{t+1}^{T-1} \ketbra{01}_{t,t+1}^c\\
\Hprop &= \frac12 \sum_{t=1}^{T-1} H_t
\end{align}
with
\begin{multline}
  H_t = \identity \otimes (\ketbra{10}_{t,t+1}^c + \ketbra{10}_{t+1,t+2}^c) \\
  - U_{t+1}	 \ketbra{110}_{t-1,t,t+1}^c  - U_{t+1}^\dagger \ketbra{100}_{t-1,t,t+1}^c
\end{multline}
where the $U_t$ correspond to the gates applied at time $t$ in the circuit being encoded.

The Hamiltonian without the output penalty,
\begin{equation}
H_0 = \Hin + \Hprop + \Hclock,
\end{equation}
has a degenerate ground space spanned by states of the form
\begin{equation}
\ket{\eta^{(0,\alpha)}} =  \frac{1}{\sqrt{T+1}}\sum_{t=0}^T \ket{\gamma_t^{({0},\alpha)}}
\end{equation}
for arbitrary $\alpha$ where
\begin{equation}
\ket{\gamma_t^{({0},\alpha)}} = \ket{\alpha_0(t)} \otimes \ket{1^t0^{T-t}}^c
\end{equation}
where $\ket{\alpha_0(t)}$ is the state of the quantum circuit at time $t$ if the input state of the ancillas and flag qubit correspond to the binary string $\textbf{0} = 0^{1+|A|}$, and the input state of the witness is given by $\ket{\alpha}$.

The modified Kitaev Hamiltonian we use is given by
\begin{equation} \label{eq:hmk}
\Hmk = \Hin + \Hprop + \kappa \Hout + \Hclock
\end{equation}
where $\kappa = 1/\poly(T) = \operatorname{o}(1/T^3)$.

\subsection{The \kham problem is \QMA-complete under faithful reductions}

Let $\mathcal{K}$ be the family of Hamiltonians of form \cref{eq:hmk}.
We begin by showing that \kham is \QMA-complete, then show that we can always choose the reductions to be faithful.

\begin{lemma}
\kham is \QMA-complete.
\end{lemma}
\begin{proof}

The proof that \kham is \QMA-complete is essentially unchanged from the proof of \QMA-completeness in~\cite{Kitaev2002}. We sketch the argument here very briefly.
Assume the circuit being encoded is a \QMA-verification circuit with completeness parameter $c$ and soundness parameter $s$.
First consider the \YES instances.
By definition, there exists a witness $w$ such that the verification circuit accepts with probability at least $c$.
It follows that the ground state of  $\Hmk$ has energy less than $\frac{\kappa (1-c)}{T+1}$.

For the \NO cases we use the following geometrical lemma.
\begin{lemma}[Geometrical lemma, Lemma 14.4~\cite{Kitaev2002}]\label{geometrical-lemma}
Let $H_1$, $H_2$ be two Hamiltonians with ground energies $a_1$, $a_2$ respectively. Suppose that for both Hamiltonians the difference between the energy of the (possibly degenerate) ground space and the next highest eigenvalue is larger than $\Lambda$, and that the angle between the two ground spaces is $\theta$. Then the ground energy of $H_1 + H_2$ is at least $a_1 + a_2 + 2\Lambda sin^2(\theta/2)$.
\end{lemma}

We apply \cref{geometrical-lemma} to $\Hmk$ with $H_1 = \Hin + \kappa \Hout$ and $H_2 = \Hprop + \Hclock$.
We have $a_1 = a_2 = 0$.
The smallest non-zero eigenvalue of $H_1$ is $\kappa$ (since $\Hin$ and $\Hout$ are commuting projectors).
The smallest non-zero eigenvalue of $H_2$ scales as $\Omega(1/T^2)$ (see~\cite{Kitaev2002} for proof).
The angle between the ground spaces satisfies
\begin{equation}
\sin^2(\theta/2) \ge \frac{1-\sqrt{s}}{4(T+1)}.
\end{equation}
Again, the proof of this is unchanged from~\cite{Kitaev2002} as the ground space of $H_1$ is equal to the ground space of $\Hin + \Hout$.

Therefore in \NO instances the ground energy of $\Hmk$ is lower bounded by $\frac{1-\sqrt{s}}{\poly(T)}$.
Setting $c-s = \Omega\left(\frac{1}{\poly(T)}\right)$ we have $\beta - \alpha = \Omega\left(\frac{1}{\poly(T)}\right)$.
Therefore we have proven a reduction from \QMA to the \lham of $\Hmk$.
\end{proof}

To show that we can always choose the reduction to be faithful, we first prove a lemma about the spectrum and low energy subspace of $\Hmk$.


\begin{lemma} \label{lem-spec-hmk}

Consider a modified Kitaev-Hamiltonian, $\Hmk$, encoding the verification circuit of some \QMA problem.
Let $Q(U)$ be the acceptance operator for a verifier circuit $U$ for some $A\in\QMA$. Set
\[
\mathcal{C}_{0} \defas\mathrm{span} \left\{\ket{\eta^{(0,\phi)}}: Q(U)\ket{\phi} = \lambda\ket{\phi}, \lambda > c  \right\}
\]
 where $c$ is the completeness parameter of the problem, and let
$g \defas c - \lambda_x$ where $\lambda_x$ is the largest eigenvalue of $Q(U)$ which is less than $c$, as in \cref{def:acc-op}.

If $g > 2T^3(T+1) \kappa$, then there exists a unitary transformation $V$ such that the subspace $\mathcal{S}_0$ defined by $\Pi_{\mathcal{S}_0} \defas V^\dagger\Pi_{\mathcal{C}_0} V$ is the low energy subspace of $\Hmk$:
\begin{equation}
\mathcal{S}_0 = \mathrm{span} \left\{\ket{\psi} : \Hmk\ket{\psi} = \lambda\ket{\psi}, \lambda  \leq \frac{\kappa (1-c)}{T+1} + T^3\kappa^2 \right\}.
\end{equation}
\begin{equation}
\| \Pi_{\mathcal{S}_0} - \Pi_{\mathcal{C}_{0}} \| = \BigO(T^3\kappa)
\end{equation}
and the spectral gap above $\mathcal{S}_0$ is given by $\Omega(\frac{g\kappa}{T+1} - T^3\kappa^2)$.

\end{lemma}
\begin{proof}
%

It is a standard result that the zero-energy ground state subspace $\mathcal{G}$ of $H_0$ is spanned by history states $\ket{\eta^{({0},\alpha)}}$ for all $\alpha$.
The spectral gap of $H_0$ is $\Omega(1/T^3)$~\cite{universality_of_adiabatic_computation}.

Since $\|\kappa \Hout \| = o\left(\frac{1}{T^3} \right) < \frac{1}{2T^3}$, the Hamiltonian $\Hmk|_\mathcal{G}$ can be approximated by the Schrieffer-Wolff perturbative expansion (see \cref{appendix-sw}).
Let $\Pi_{\mathcal{G}}$ be the projector onto $\mathcal{G}$.

The zeroth order term in the expansion is given by $H_0\Pi_{\mathcal{G}} = 0$.
The matrix elements of the first order term, $\Pi_{\mathcal{G}} \Hout \Pi_{\mathcal{G}}$, are given by~\cite{deshpande2020importance}~[Appendix B]:
\begin{equation}
\bra{\eta^{(0,\alpha)}} \Pi_{\mathcal{G}} \Hout \Pi_{\mathcal{G}} \ket{\eta^{(0,\beta)}} = \frac{\kappa}{T+1}\left( \bra{\alpha}\ket{\beta} - \bra{\alpha}Q\ket{\beta}\right)
\end{equation}

Denote the eigenstates of $Q(U)$ by $\ket{\phi_1}, \ket{\phi_2}, \cdots, \ket{\phi_{2^w}}$ with associated eigenvalues $\lambda_1 \geq \lambda_2 \geq \cdots \geq \lambda_{2^w}$.
In the basis spanned by $\ket{\eta^{(0,\phi_i)}}$, the first order term in the Schrieffer-Wolff expansion is diagonal:
\begin{equation} \label{Eq:first_order_diag}
\Pi_{\mathcal{G}}\Hout \Pi_{\mathcal{G}} = \frac{\kappa}{T+1} \sum_i(1-\lambda_i) \ket{\eta^{(0,\phi_i)}}\bra{\eta^{(0,\phi_i)}}
\end{equation}
By \cref{H-bound-error} and \cref{Eq:first_order_diag} we conclude that in $\mathcal{G}$ the eigenvalues of $\Hmk$ are given by
\begin{equation}
\tilde{\lambda}_i = \frac{\kappa (1-\lambda_i)}{T+1} \pm T^3\kappa^2
\end{equation}

%

Let $\mathcal{R} \defas \mathrm{span}\left\{\ket{\psi} :\Hmk\ket{\psi} = E\ket{\psi}, E \leq \frac{\kappa}{2} \right\}$.
By the $\sin(\theta)$ theorem~\cite{davis1969}:
\begin{equation}
\|\Pi_\mathcal{R} - \Pi_{\mathcal{G}} \| = \|V^\dagger\Pi_\mathcal{G} V- \Pi_{\mathcal{G}} \| \leq T^3 \kappa
\end{equation}
where $V = e^S$ is the Schrieffer-Wolff transformation. Therefore
\begin{equation}
\|V^\dagger\Pi_{\mathcal{C}_0} V- \Pi_{\mathcal{C}_0} \|= \|\Pi_{\mathcal{S}_0} - \Pi_{\mathcal{C}_0} \|  \leq T^3 \kappa.
\end{equation}

$\mathcal{C}_0$ is spanned by history states satisfying $\Heff(1)\ket{\eta^{(0,\phi)}} = \lambda \ket{\eta^{(0,\phi)}}$ for $\lambda \leq \frac{\kappa (1-c)}{T+1}$.
The corresponding eigenvalues of $\Hmk$ are upper bounded by $\frac{\kappa (1-c)}{T+1} + T^3 \kappa^2$.
The smallest eigenvalue of $\Hmk$ which is larger than $\frac{\kappa (1-c)}{T+1} + T^3 \kappa^2$ is lower bounded by $\frac{\kappa (1-\lambda_x)}{T+1} - T^3\kappa^2$.

Therefore, since $g>2T^3(T+1)\kappa$, the subspace
\begin{equation}
\mathcal{S}_0= \mathrm{span} \left\{\ket{\psi} : \Hmk\ket{\psi} = \lambda\ket{\psi}, \lambda  \leq \frac{\kappa (1-c)}{T+1} + T^3\kappa^2 \right\}
\end{equation}
and the spectral gap of $\Hmk$ above $\mathcal{S}_0$ is given by $\Omega(\frac{g\kappa}{T+1} - T^3\kappa^2)$.
\end{proof}

\begin{lemma} \label{lem:kh-faithful}
The \kham problem is \QMA-complete under faithful reductions.
\end{lemma}
\begin{proof}
For any verification circuit, $U$, of any problem in \QMA, we can require that the computation `idles' in its initial state for $L$ time steps before carrying out its verification computation (``idling to enhance coherence''~\cite{universality_of_adiabatic_computation}).


The history state of the computation for the first $L$ time steps will be given by
\begin{equation}
\ket{\eta_\mathrm{idling}^{(0,\alpha)}} = \ket{\alpha} \otimes \ket{0}^{|A|+1} \otimes \frac{1}{\sqrt{L}}\sum_{t=1}^L\ket{1^t0^{T-t}}^c.
\end{equation}
The rest of the history state is captured in
\begin{equation}
\ket{\eta_\mathrm{comp}^{(0,\alpha)}} = \frac{1}{\sqrt{T-L}}\sum_{t=L+1}^T\ket{\gamma_t^{(0,\alpha)}}
\end{equation}
So the total history state is given by
\begin{equation}
\ket{\eta^{(0,\alpha)}} = \sqrt{\frac{L}{T}}\ket{\eta_\mathrm{idling}^{(0,\alpha)}}  + \sqrt{\frac{L-T}{T}}\ket{\eta_\mathrm{comp}^{(0,\alpha)}}
\end{equation}

The encoding $\mathcal{E}(M) = VMV^\dagger$ defined via the isometry
\begin{equation}
V = \sum_i \ket{\eta_\mathrm{idling}^{(0,i)}}\bra{i},
\end{equation}
where the $\ket{i}$ are computational basis states, is local. (This can be verified by direct calculation, see~\cite{kohler2020translationallyinvariant}.)


Moreover, we have that
\begin{equation}
\label{eq:V-V'}
\begin{split}
\|\Pi_{\mathcal{C}_0}- \Pi_{\mathcal{E}(\mathcal{L})}\|^2 & = \left\| \sum_{\ket{\phi} \in \mathcal{L}} \left(\ket{\eta^{(0,\phi)}}\bra{\eta^{0,\phi}} - \ket{\eta_{\mathrm{idling}}^{{(0,\phi)}}}\bra{\eta_{\mathrm{idling}}^{{(0,\phi)}}}\right) \right\| ^2\\
 & = \left\| \sum_{\ket{\phi} \in \mathcal{L}} \left(\sqrt{\frac{T -L}{T}}\ket{\eta_{\mathrm{comp}}^{{(0,\phi)}}}\bra{\eta_{\mathrm{comp}}^{{(0,\phi)}}} + \left(\sqrt{\frac{L}{T}}-1\right)\ket{\eta_{\mathrm{idling}}^{{(0,\phi)}}}\bra{\eta_{\mathrm{idling}}^{{(0,\phi)}}} \right) \right\|^2 \\
 &  \leq 2\left(1-\sqrt{\frac{L}{T}}\right)
\end{split}
\end{equation}
where $\mathcal{L} \defas \{\ket{\phi}:Q\ket{\phi} = \lambda\ket{\phi}, \lambda > c \}$
Therefore $\|\Pi_{\mathcal{C}_0}- \Pi_{\mathcal{E}(\mathcal{L})}\|$ can be made arbitrarily small by increasing $L$.

The result follows immediately from \cref{lem-spec-hmk} and the triangle inequality.
\end{proof}

\section{General Conditions for Universality} \label{general-condition}


In order to state our main theorem we require one more definition.

\begin{definition}[Closed Hamiltonian model]
We say that a model $\mathcal{M}$, is closed if for any pair of Hamiltonians $H_A^{(1)}, H_B^{(2)} \in \mathcal{M}$ acting on sets of qudits $A,B$ respectively where in general $A \cap B \neq \{\}$, there exists a Hamiltonian $H^{(3)} \in \mathcal{M}$ which can efficiently simulate $H_A^{(1)} + H_B^{(2)}$.
\end{definition}

We can now prove our main result \cref{th:classification_intro}:

\classification*

\begin{proof}
First consider the only if direction.
Closure is clearly necessary: if a model $\mathcal{M}$ is universal, all Hamiltonians (including those of the form $H_A^{(1)} + H_B^{(2)}$ for $H_A^{(1)}, H_B^{(2)} \in \mathcal{M}$) can be simulated by a Hamiltonian in the model.
In \cref{lem:kh-faithful} we proved the \kham problem is \QMA-complete under faithful reductions.
Any efficient universal model must be able to simulate Hamiltonians in $\mathcal{K}$ with only polynomial overhead, hence \mham must itself be \QMA-complete under faithful reductions.

Now consider the if direction.
Let $\mathcal{M}$ be a family of Hamiltonians meeting the conditions of the theorem, i.e.\ such that \mham is \QMA-complete under faithful reductions, and $\mathcal{M}$ is closed.
We will explicitly construct a universal model, solely based on these conditions.

Consider the following problem:
\begin{problem}[\yesHam] \label{yesHam}
\probleminput{
A $k$-local Hamiltonian $\Htarget$ acting on $n$ spins with local dimension $d$.
}
\problemquestion{Output
\YES\
}
\end{problem}

This problem is (clearly) trivial.
But we can choose to construct a non-trivial \QMA verification circuit for it.
We will choose a verification circuit which picks out a particular subspace that allows us to prove universality.
By \cref{def:qma-under-faithful} there must be a faithful reduction with respect to \emph{this} verification circuit from \yesHam to \mham.

The verification circuit we choose, and the subspace it picks out, are captured in the following.

\begin{lemma} \label{lem:yesHam}
\yesHam can be verified by a circuit $U_a$ with gapped acceptance operator $Q(U_a)$ with ground space
\begin{equation}
\mathcal{L}_0 \defas \mathrm{span}\{\ket{\phi}: Q(U_a)\ket{\phi} = \ket{\phi} \}
\end{equation}
satisfying
\begin{equation}
\|\Pi_{\mathcal{L}_0} - \Pi_{\mathcal{W}} \| \leq \BigO\left(a^{-1}  \right)
\end{equation}
where
\begin{equation}
\mathcal{W} \defas \mathrm{span}\left\{\ket{w_\mu} = \frac{1}{\sqrt{a^2+1}} \ket{\psi_\mu}\left(a\ket{\#} + \ket{E_\mu}\right): \Htarget \ket{\psi_\mu} = E_\mu\ket{\psi_\mu} \right\}.
\end{equation}
\end{lemma}

\begin{proof}
The verifier circuit, $C_V$, acts on the witness and two ancilla registers, $B$, $B'$.
It will be helpful to divide the witness into two separate registers: An $A$ register, which is $n$ $d$-dimensional qudits.
And an $A'$ register, which consists of $m$ qutrits with orthonormal basis states $\ket{\#}$, $\ket{0}$ and $\ket{1}$, where $m = \log_2(\epsilon)$.
The $B$ register is the same size as the $A'$ register.
The $B'$ register consists of a single qubit.

The verifier $C_V$ operates as follows:
\begin{enumerate}
\item Apply a unitary rotation $P_a: \ket{0} \rightarrow \frac{1}{\sqrt{a^2+1}} (a\ket{\#} + \ket{1})$ to the $B'$ register.f
\item Carry out controlled-phase-estimation on the $A$ register with respect to the unitary generated by $\Htarget$, $U = e^{i\Htarget\tau}$, for some $\tau$ such that $\|\Htarget \tau \|<2\pi$. The $B'$ register serves as the control qubit. Calculate (an approximation to) the energy $E_\mu$ from the eigenphase $\theta_\mu$ and store the result in the $B$ register (in binary).

Calculating $E_\mu$ to accuracy $\epsilon$ requires calculating the eigenphase $\theta_\mu$ to accuracy $\BigO(\epsilon/\|\Htarget\|)$ which takes $\BigO(\|\Htarget\|/\epsilon)$ uses of $U = e^{i\Htarget\tau}$. The unitary $U$ must therefore be implemented to accuracy $O(\epsilon/\|\Htarget\|)$, which can be done with overhead $\poly(n,d^k, \|\Htarget\|,\tau,1/\epsilon)$ where $n$ is system size, $d$ is local dimension and $k$ is locality via~\cite[Lemma 3.3]{kohler2020translationallyinvariant}. The whole procedure takes time $\TPE = \poly(n,d^k,\|\Htarget\|/\epsilon)$
\item Carry out a SWAP test \cite{2001} between registers $A'$ and $B$. Accept if outcome $0$ is measured, reject otherwise.
\end{enumerate}
The entire procedure takes time $T = \BigO(\poly(n,d^k, \|\Htarget \|)/\epsilon)$.

Let
\begin{equation}
\ket{\alpha_\mu} = \sum_j \left(\frac{1}{2^m} \sum_{k=0}^{2^m-1} e^{2\pi i k (E_\mu - j/2^m)}\right) \ket{j}
\end{equation}
be the result of applying the phase estimation algorithm on $\ket{\psi_\mu}$ with respect to $U = e^{i\Htarget\tau}$.
Then evidently
\begin{equation}
\ket{\phi_\mu} = \frac{1}{\sqrt{a^2+1}} \ket{\psi_\mu}_A(a\ket{\#}_{A'} + \ket{\alpha_\mu}_{A'})
\end{equation}
is an eigenvector of $Q(U_a)$ with eigenvalue 1, and all eigenvectors of $Q(U_a)$ with eigenvalue 1 are in $\mathrm{span}\{ \ket{\phi_{\mu}} \}$.

Moreover,
\begin{equation}
\bra{w_\mu}\ket{\phi_\mu} \leq \frac{a^2 + \frac{4}{\pi^2}}{a^2+1}
\end{equation}

Therefore
\begin{equation}
\|\ketbra{w_\mu} - \ketbra{\phi_\mu} \| \leq 2 \sqrt{1 - \left( \frac{a^2 + \frac{4}{\pi^2}}{a^2+1}\right)^2}
\end{equation}
and
\begin{equation} \label{eq:subspaces-a}
\|\Pi_{\mathcal{L}_0} - \Pi_{\mathcal{W}} \| \leq \BigO\left(a^{-1} \right).
\end{equation}

The next largest eigenvalue of $Q(U_a)$ is $\frac{1}{2}$.
\end{proof}

It follows immediately from \cref{eq:subspaces-a} that
\begin{equation}
\|\mathcal{E}(\Pi_{\mathcal{L}_0}) - \mathcal{E}(\Pi_{\mathcal{W}}) \| \leq \BigO\left(a^{-1} \right)
\end{equation}
for any encoding $\mathcal{E}$.

It follows from the triangle inequality and \cref{faithfulness_definition}(\labelcref{first_faithfulness_condition}) that for any instance of \yesHam there exists $\HLS \in \mathcal{M}$ with low energy subspace $\mathcal{S}_0 \defas \mathrm{span}\{\ket{\phi} : H_x\ket{\phi} = \lmin \ket{\phi} \}$ such that
\begin{equation} \label{eq:subspaces}
\|\Pi_{\mathcal{S}_0} - \Pi_{\mathcal{E}\left(\mathcal{W}\right)}\| \leq \eta + \BigO\left(a^{-1}  \right)
\end{equation}
where $\mathcal{E}=V\left(M \otimes P + \overline{M}\otimes Q \right)V^\dagger$ is some local encoding and $\eta$ can be chosen to be arbitrarily small.
The spectral gap above $\mathcal{S}_0$ is $\Omega(1/\poly(T))$.

Another trivial problem (that is therefore also evidently in \QMA) is:
\begin{problem}[\flagid]
\probleminput{
Classical description of a one-qudit state $\ket{f}$}

\problemquestion{
Output \YES.
}
\end{problem}

For this problem we will use a faithful reduction with respect to the non-trivial verification circuit which simply measures a single qudit in the $f$ basis.
So, for any single qudit state $\ket{f}$, there exists $H_f \in \mathcal{M}$ such that:
\begin{itemize}
\item $H_{f} \ket{\mathcal{E}_{\mathrm{state}}(\phi)} = \lambda^{(f)}_0  \ket{\mathcal{E}_{\mathrm{state}}(\phi)}$ for all $\ket{\phi}$ such that $\bra{\phi}\ket{f} = 0$
\item $H_{f} \ket{\mathcal{E}_{\mathrm{state}}(f)} = \lambda^{(f)}_1 \ket{\mathcal{E}_{\mathrm{state}}(f)}$
\end{itemize}
for some local encoding ${\mathcal{E}_{\mathrm{state}}}$, where $\lambda^{(f)}_k$  can be efficiently computed. (Since the problem size is $O(1)$ for a state $\ket{f}$ that can be described in $O(1)$ bits.) Wlog we will take $\lambda^{(f)}_k = k$.

Consider a Hamiltonian acting on $N$ spins:
\begin{equation}
\Hsim = \Delta \left( \HLS - \lmin \identity \right) + a \sum_{i=n'+1}^N 2^{i-(n'+1)} H_i^{(1)}
\end{equation}

Where $\HLS \in \mathcal{M}$ is a faithful reduction (with respect to the verifier defined in \cref{lem:yesHam}) from \yesHam for the Hamiltonian  $\Htarget = \sum_\mu E_\mu \ketbra{\psi_\mu}$.
$E_\mu$ in $\ket{w_\mu}$ is expressed in binary to precision $\epsilon$ in qudits $[n'+1,N]$.\footnote{Here $n'$ is the number of spins in the encoded witness state $\mathcal{E}\left( \ket{w_\mu}\right)$. Since this is a reduction to \QMA we have $n' = \BigO(\poly(n))$.}
The $H_i^{(1)} \in \mathcal{M}$ are faithful reductions (with respect to the obvious verification circuit) from \flagid for the flag states $\ket{1}$ acting on the $i^{\mathrm{th}}$ qudit.
We will require $\Delta > \|\Htarget\|$.

First we show that $\Hsim$ can simulate $H' = \sum_\mu E_\mu \ketbra{w_\mu}$.
The low energy subspace of $\HLS$ consists of states in the subspace $\mathcal{S}_0$.
On states in $\mathcal{S}_0$, $\HLS - \lmin$ has energy zero.
While on states in $\mathcal{E}(\mathcal{W})$, $a \sum_{i=n'+1}^N 2^{i-(n'+1)} H_i^{(1)}$ has energy in the range $[E_\mu - \epsilon, E_\mu+\epsilon]$.

It follows from~\cite[Lemma 24]{Cubitt:2017} and \cref{eq:subspaces} that there exists an encoding $\mathcal{E}'(M) = V'\left(M \otimes P + \overline{M}\otimes Q \right)V'^\dagger$ such that
\begin{equation}\label{eq:isom-diff}
\|V'-V\| \leq \sqrt{2}(\eta + O(a^{-1}))
\end{equation}
and $\mathcal{E}'(\identity) = \Pi_{\mathcal{S}_0}$.
Moreover,
\begin{equation}
V' = WV
\end{equation}
where $W$ is a unitary satisfying
\begin{equation}
\Pi_{\mathcal{S}_0} = W\Pi_{\mathcal{E}(\mathcal{W})}W^\dagger
\end{equation}
and
\begin{equation}
\| W - \identity \| \leq \sqrt{2} \|\Pi_{\mathcal{S}_0} -\Pi_{\mathcal{E}(\mathcal{W})}  \| \leq \BigO(\eta + a^{-1})
\end{equation}

We need the following technical lemma.

\begin{lemma}[First-order simulation~\cite{BH17}[Lemma 14]] \label{technical-lemma}
	Let $H_0$ and $H_1$ be Hamiltonians acting on the same space and $\Pi$ be the projector onto the ground space of $H_0$. Suppose that $H_0$ is zero on $\Pi$ and the next smallest eigenvalue is at least 1.
	Let $U$ be an isometry such that $UU^{\dagger}=\Pi$ and
	\begin{equation}
	\label{eq:firstorderrequirement}
	\|U \Htarget U^\dag - \Pi H_1 \Pi\| \le \epsilon/2.
	\end{equation}
	Let $H_{\operatorname{sim}} = \Delta H_0 + H_1$ .
	Then there exists an isometry $\tilde{V}$ onto the the space spanned by the eigenvectors of $H_{\operatorname{sim}}$ with eigenvalue less than $\Delta/2$ such that
	\begin{enumerate}
		\item $\|U-\tilde{V}\| \le \BigO\left( \Delta^{-1}\| H_1\|  \right)
$
		\item $\|\tilde{V}H_{\operatorname{target}} \tilde{V}^{\dagger} -H_{\operatorname{sim}< \Delta/2} \| \le \BigO\left( \Delta^{-1}\| H_1\|^{2}\right) + \epsilon/2 $
	\end{enumerate}
\end{lemma}
We will apply \cref{technical-lemma} with $H_1=a \sum_{i=n'+1}^N 2^{i-(n'+1)} H_i^{(1)}$ and $H_0= \delta \HLS$ where $\delta = \BigO(\poly(T))$.
We have that in $\Pi_{\mathcal{S}_0}$, $\HLS$ has energy zero and by \cref{faithfulness_definition}(\labelcref{second_faithfulness_condition}) the spectral gap above $\mathcal{S}_0$ scales as $\Omega(1/\poly(T))$ so $H_0=\delta\HLS$ has next smallest eigenvalue at least 1.

Moreover, $\left\|H_1\right\| =a\|\Htarget \|$.
Note that $V'$ is an isometry which maps onto the ground state of $H_0$, $\mathcal{S}_0$.
By construction we have that the spectrum of $\Htarget$ is approximated to within $\epsilon $ by $H_1$ restricted to $\mathcal{E}(\mathcal{W})$, so $\|\Pi_{\mathcal{E}(\mathcal{W})} H_1 \Pi_{\mathcal{E}(\mathcal{W})}  - {\mathcal{E}}(\Htarget)\| \leq \epsilon$.

Using that the operator norm is unitarily invariant, and that $V' = WV$ gives
\begin{equation}
\|W \Pi_{\mathcal{E}(\mathcal{W})} H_1 \Pi_{\mathcal{E}(\mathcal{W})}W^\dagger  - {\mathcal{E'}}(\Htarget)\| \leq \epsilon.
\end{equation}

We also have
\begin{equation}
\begin{split}
\|\Pi_{\mathcal{S}_0} H_1 \Pi_{\mathcal{S}_0} - W \Pi_{\mathcal{E}(\mathcal{W})} H_1 \Pi_{\mathcal{E}(\mathcal{W})}W^\dagger\| & =  \|\Pi_{\mathcal{S}_0} H_1 \Pi_{\mathcal{S}_0} - \Pi_{\mathcal{S}_0} WH_1W^\dagger \Pi_{\mathcal{S}_0}\| \\
& \leq  \|H_1 - WH_1 W^\dagger \| \\
& \leq 2 \|H_1\| \|\identity-W\| \\
& \leq \BigO\left(a \eta \|\Htarget \| \right)
\end{split}
\end{equation}
where we have used~\cite[Lemma 18]{Cubitt:2017} in the penultimate step.
So
\begin{equation}
\|\Pi_{\mathcal{S}_0} H_1 \Pi_{\mathcal{S}_0} - \mathcal{E}'(\Htarget) \| \leq \epsilon + \BigO\left(a \eta \|\Htarget \| \right).
\end{equation}

\cref{technical-lemma} therefore implies that there exists an isometry $\tilde{V}$ that maps exactly onto the low energy space of $\Hsim$ such that $\|\tilde{V}-V'\|\le \BigO\left(\|\Htarget \|a/(\Delta/\delta)\right )=\BigO\left(a \delta \|\Htarget \|/\Delta \right)$.
By the triangle inequality and \cref{eq:isom-diff}, we have
\begin{equation}\label{eq:T-scaling}
\|V-\tilde{V}\|\le \|V-V'\|+\|V'-\tilde{V}\|\le \BigO \left(\frac{a\poly(T')\|\Htarget \|}{\Delta} + \eta +a^{-1}  \right).
\end{equation}

The second part of the lemma implies that
\begin{align}
  \|\tilde{V} H' \tilde{V}^{\dagger} -(\Hsim)_{<\Delta'/2}\|
  &\le \epsilon + \BigO\left(a \eta \|\Htarget \| + (a|\Htarget \|)^2/(\Delta/\delta) \right)\\
  &=  \epsilon  +\BigO\left(a \eta \|\Htarget \|+ \frac{a^2|\Htarget \|^2 \delta}{\Delta}   \right).
\end{align}
Therefore, the conditions of \cref{app-sim} are satisfied for a $(\Delta',\eta',\epsilon')$-simulation of $H'$, with $\eta' = \BigO \left(\frac{a \poly(T')\|\Htarget \|}{\Delta} + \eta + a^{-1} \right)$, $\epsilon' = \epsilon   +\BigO\left(a \eta \|\Htarget \|+ \frac{a^2\poly(T')\|\Htarget \|^2}{\Delta} \right)$ and $\Delta'= \Delta/\delta = \Delta / \poly(T)$.

By definition we can choose $\eta$ to be arbitrarily small.
We can also make $\BigO\left(a^{-1} \right)$ arbitrarily small.
By increasing $T$, we can also make $\epsilon$ arbitrarily small.
Therefore, by choosing $\Delta$ such that
\begin{equation}
\Delta \geq  \Delta'\poly(T') + \frac{a \poly(T')\|\Htarget \|}{\eta'} + \frac{a^2\poly(T')\|\Htarget \|^2}{ \epsilon' }
\end{equation}
we can construct $\Hsim$ which is a $(\Delta', \eta', \epsilon')$-simulation of $H'$ with arbitrarily small $\epsilon'$, $\eta'$.
Since $\Hsim$ is a sum of Hamiltonians which are all in $\mathcal{M}$, by the closure property there exists $\Huniv \in \mathcal{M}$ which can efficiently simulate $\Hsim$.
Therefore, since simulations compose \cite[Lemma~17]{Cubitt:2017} $\Huniv$ can simulate $H'$.

Finally, we show that $H' = \sum_\mu E_\mu \ketbra{w_\mu}$ is itself a simulation of $\Htarget$. Consider the local encoding
\begin{equation}
\mathcal{E'}(M) = W M W^\dagger,
\end{equation}
where $W = \sum_\mu \ket{\psi_\mu}\ket{0} \bra{\psi_\mu}$, and the non local encoding
\begin{equation}
\tilde{\mathcal{E}'}(M) = \tilde{W} M \tilde{W}^\dagger
\end{equation}
withk
\begin{equation}
  \tilde{W} = \frac{1}{\sqrt{a^2+1}}  \sum_\mu \ket{\psi_\mu} \left(a \ket{\#}+ \ket{E_\mu} \right) \bra{\psi_\mu}.
\end{equation}
We have that
\begin{equation}
\|W - \tilde{W}\| = 2\left(1 - \frac{a}{\sqrt{a^2+1}} \right)
\end{equation}
So by increasing $a$ we can make the norm arbitrarily small. We also have that $S_{\tilde{\mathcal{E}}'} = S_{H'}$, so condition~(\labelcref{app-sim:1}) from \cref{app-sim} is met. The spectrum of $H'$ is exactly the spectrum of $\tilde{\mathcal{E}}'(\Htarget)$, so condition~(\labelcref{app-sim:2}) of \cref{app-sim} is also met. Therefore $H'$ is a simulation of $\Htarget$.

Using the composition of simulations again, we have that $\Huniv$ can simulate $\Htarget$.
We have left $\Htarget$ arbitrary, so $\mathcal{M}$ is a universal model.

Finally we consider efficiency.
The simulation of $\Htarget$ by $H'$ is clearly efficient.
To see that the simulation of $H'$ by $\Hsim$ is efficient, note that the number of qudits in the simulation, $N$, must be polynomial in $n$ and $\|\Htarget\|$ as $\HLS$ is in \QMA.
Furthermore, $\|\Hsim\| = \Omega(\Delta) = \poly(T',\|\Htarget\|,1/\epsilon',1/\eta') = \poly(n,\|\Htarget\|,1/\epsilon',1/\eta')$.
Thus $\Hsim$ is an efficient simulation of $H'$.
\end{proof}

There are two corollaries about universal Hamiltonians which aren't efficient in the sense of~\cite{Cubitt:2017}, but which are nonetheless interesting universal models which are better suited to some applications.

\begin{definition}
We say a Hamiltonian can be described succinctly if it can be described by $\BigO(\log(n))$ bits of information when acting on $n$ qudits.
\end{definition}

\begin{corollary} \label{co:succcinctClassification}
Let $\mathcal{M}$ be a family of succinct Hamiltonians. Then $\mathcal{M}$ is universal and can efficiently simulate any succinct Hamiltonian iff \mham is \QMAEXP-complete under faithful reductions and $\mathcal{M}$ is closed.
\end{corollary}
\begin{proof}
This proof relies on the same ideas as \cref{th:classification_intro}, so here we sketch the main ideas, highlighting where the proofs differ.
First, note the proof that \kham is \QMA-complete under faithful reductions (\cref{lem:kh-faithful}) can be repurposed to prove that the \mham which was shown to be \QMAEXP-complete in~\cite{gottesman2009} remains \QMAEXP complete under faithful reductions.
The only if direction follows immediately as in \cref{th:classification_intro} (where now the overhead can be exponential).

To see the if direction, we will assume there exists a family of Hamiltonians, $\mathcal{M}$, such that \mham is \QMAEXP-complete under faithful reductions and $\mathcal{M}$ is closed.
Consider the following computational problem:

\begin{problem}[\succinctYesHam] \label{succinctYesHam}
\probleminput{
A $k$-local Hamiltonian $\Htarget$ acting on $n$ spins with local dimension $d$, which can be described succinctly.}
\problemquestion{Output
\YES\
}
\end{problem}

\noindent \succinctYesHam is clearly a trivial problem.
But, as with \yesHam, we can construct a non-trivial verification circuit, which picks out a particular subspace that allows us to prove universality.
By \cref{def:qma-under-faithful} there must be a faithful reduction with respect to \emph{this} verification circuit from \succinctYesHam to \mham.

The verification circuit we choose, and the subspace it picks out, are as in \cref{lem:yesHam}.
Although the circuit is unchanged, this is now a \QMAEXP-verification circuit, since for Hamiltonians acting on $n$ qudits the circuit length and witness size are of order $\poly(n) = \BigO(2^{\poly(x)})$, where $x$ is the number of bits of information needed to describe the input to the problem.

It then follows, using the same argument as in \cref{lem:yesHam}, that for any Hamiltonian $\Htarget$ acting on $n$ spins which can be described succinctly, there exists a Hamiltonian in $\mathcal{M}$ which can simulate $\Htarget$ efficiently (where efficiency is defined in terms of number of qudits, not bits of information).

To prove universality, note that in~\cite[Theorem 3.6]{kohler2020translationallyinvariant} a construction is given of a universal Hamiltonian, $\Htiuniv$, (with exponential overhead in terms of number of spins and norm of simulating system) which can be described succinctly.

Since there exists a Hamiltonian in $\mathcal{M}$ which can simulate $\Htiuniv$ (for any values of the parameters in $\Htiuniv$), and since simulations compose, it follows that $\mathcal{M}$ is a universal model.
When simulating general (non-succinct) Hamiltonians, the universal model, $\mathcal{M}$, inherits an exponential overhead in terms of the numbers of qudits and the norm of the simulating system from $\Htiuniv$.
\end{proof}

\QMAEXP is a more powerful complexity class than \QMA so it may seem odd that it appears to be less efficient as a simulator.
However, there are some situations where using a family of Hamiltonians meeting the conditions of \cref{co:succcinctClassification} will give a more efficient simulator than using a family of Hamiltonians meeting the conditions of \cref{th:classification_intro}.
To see this, note that given a Hamiltonian which can be described succinctly, it can be simulated efficiently (in the sense of~\cite{Cubitt:2017} i.e.\ in terms of numbers of qudits, simulating system norm, and $\epsilon$ and $\eta$ parameters) by either a family of Hamiltonians with a \QMA-complete \mham, or a family of Hamiltonians with a \QMAEXP-complete \mham.
However, the simulation by the \QMA-complete family of Hamiltonians will not be efficient in terms of the number of bits needed to describe the simulating Hamiltonian.
Whereas the simulation by the \QMAEXP-complete family of Hamiltonians would be.
So there are situations where simulation using a \QMAEXP-complete family of Hamiltonians is more efficient, demonstrating that the question of which family of Hamiltonians is a `more powerful' simulator doesn't have a straightforward answer.

The obvious example of Hamiltonians that can be succinctly described are translationally invariant Hamiltonians.
Examples of translationally invariant universal Hamiltonians are constructed in~\cite{PiddockBausch,kohler2020translationallyinvariant}, where it is noted that  a translationally invariant universal model with fixed interactions must have an exponential overhead in terms of number of spins by a simple counting argument.

By considering the problem \precmham and introducing the idea of a \emph{exponentially faithful}, reduction we can also derive conditions for universal models which are efficient in terms of the number of qudits in the simulator system, but not in terms of the simulating system's norm.
We say a reduction is \emph{exponentially faithful} if it meets the conditions of \cref{faithfulness_definition}, but where now the gap in the spectrum of the acceptance operator, $g$, and the corresponding gap in the spectrum of the Hamiltonian is required to satisfy $g > 1/\exp(n)$ for $n$ the size of the input.
This is a natural relaxation when considering \precmham, as it requires the gap in the spectrum of the acceptance operator in \YES cases to scale in the same way as the promise gap for the problem.

The natural complexity class when considering \precmham is \PreciseQMA.
It is known that every problem in \PreciseQMA can be solved by a quantum circuit of length $T = \BigO(\exp(n))$ acting on $\poly(n)$ qudits, with completeness $c = 1 - 2^{\poly(n)} $ and soundness $s = 2^{\poly(n)} $ (where $n$ is the size of the problem input) \cite[Corollary~10]{Fefferman2016}.
Therefore, when defining what it means for \precmham to be \PreciseQMA-complete under exponentially faithful reductions there are two classes of circuits we could require faithfulness with respect to - the polynomial sized circuits that give exponentially small completeness-soundness gap, or the exponential sized circuits that give completeness-soundness gap exponentially close to 1.
Here we choose the latter, and define:

\begin{definition} \label{def:precqma-under-faithful}
We say that \precmham is \PreciseQMA-complete under exponentially faithful reductions if for all $A \in $ \PreciseQMA and for any exponential time verification circuit $U$ which verifies $A$, there exists a reduction from $A$ to the \precmham problem which is exponentially faithful with respect to $U$.
\end{definition}

\begin{corollary}\label{cr:precise}
A family of Hamiltonians, $\mathcal{M}$, is a universal model which is
\begin{enumerate}[1.]
\item efficient in terms of the numbers of qudits,
\item  not efficient in terms of the simulating system's norm and
\item achieves exponential accuracy in the $\epsilon$ parameter with polynomial overhead in number of qudits and exponential overhead in simulating system norm
\end{enumerate}
iff
\begin{enumerate}[i.]
\item \precmham is \PSPACE-complete under exponentially faithful reductions,
\item $\mathcal{M}$ is closed and
\item \mham is not \QMA-complete under faithful reductions.
\end{enumerate}
\end{corollary}
\begin{proof}
Recall that \PSPACE = \PreciseQMA~\cite{Fefferman2016}.

First consider the if direction.
Assume we have a family of Hamiltonians $\mathcal{M}$ meeting the conditions $i-iii$ of the theorem.
Since \mham is not \QMA-complete under faithful reductions we have to use exponentially faithful reductions to \precmham, where by \cref{def:precqma-under-faithful} we are considering faithfulness with respect to exponential time circuits.
The if direction follows immediately since going through the proof of \cref{th:classification_intro} with an exponentially small gap in the acceptance operator and exponentially long computation time requires an exponentially large energy penalty $\Delta$, and gives an exponentially small accuracy parameter $\epsilon$.

To see the only if direction, consider a universal model, $\mathcal{M}$ meeting conditions $1-3$ of the theorem.
Necessity of closure is trivial.
Note, \mham cannot be \QMA-complete under faithful reductions, because by \cref{th:classification_intro} if it was the universal model would be efficient in terms of norm.

Consider simulating the family of Hamiltonians $\mathcal{K}$ using the model $\mathcal{M}$ to exponential accuracy in $\epsilon$.
This demonstrates that \precmham  is \PSPACE-complete (including under exponentially faithful reductions), but it does not contradict the statement that \mham is not \QMA-complete.
This is because the \lham problem requires that the terms in the Hamiltonian are of order 1, which requires dividing each term in the simulator system by the simulating system norm, which by assumption is exponential in the size of the system.
This leads to a Hamiltonian with an exponentially small spectral gap, which attenuates the promise gap too fast to maintain \QMA-completeness, but gives \PreciseQMA-completeness (and therefore \PSPACE-completeness).

\end{proof}

An example of a universal Hamiltonian meeting the conditions of \cref{cr:precise} is given in~\cite{kohler2020translationallyinvariant}.
It is a translationally invariant universal model, but includes a phase parameter which encodes information about the target system, so the interactions are not fixed.

%
%
%
%

\appendix

\section{The Schrieffer-Wolff expansion} \label{appendix-sw}

Consider a finite dimensional Hilbert space decomposed into a direct sum:
\begin{equation}
\mathcal{H} = \mathcal{H}_+ \oplus \mathcal{H}_-
\end{equation}

Let $\Pi_-$ be the projector onto $\mathcal{H}_-$ and $\Pi_+$ be the projector onto $\mathcal{H}_+$.
Let $H_0$ and $H_1$ be Hermitian operators acting on $\mathcal{H}$ such that $H_0$ is block diagonal with respect to the direct sum.
Assume all the eigenvalues of $H_0$ on $\mathcal{H}_-$ are in the range $[0,\lambda_0]$ for $\lambda_0 < 1$.

Consider the perturbed Hamiltonian $\tilde{H} = \Delta H_0 + H_1$ where $\Delta \gg 1$ and $\| V\| < \frac{\Delta}{2}$.
The Schrieffer-Wolf transformation is a unitary rotation $e^S$ which is used to perturbatively diagonalise $\tilde{H}$.
It satisfies the following properties:
\begin{equation}
\Pi_- \left(e^S\tilde{H} e^{-S} \right)\Pi_+ = \Pi_+ \left(e^S\tilde{H} e^{-S} \right)\Pi_- = 0
\end{equation}

\begin{equation}
\Pi_- S\Pi_- = \Pi_+ S\Pi_+ = 0 \textrm{      ,      } \|S\| <\frac{\pi}{2}
\end{equation}

The effective low-energy Hamiltonian $\Heff$ acting on $\mathcal{H}_-$ is given by
\begin{equation}
\Heff = \Pi_- \left(e^S \tilde{H} e^{-S} \right)\Pi_-
\end{equation}

Define $\mathcal{R} \defas \left\{\ket{\psi} : \bra{\psi}\tilde{H} \ket{\psi} \in \left[\lambda_0 - \frac{\Delta}{2}, \lambda_0 + \frac{\Delta}{2} \right] \right\}$, and let $\Pi_\mathcal{R}$ be the projector onto $\mathcal{R}$.
Then
\begin{equation} \label{Eq:SW-projectors}
e^S\Pi_\mathcal{R} e^{-S} = \Pi_-
\end{equation}

The operators $S$ and $\Heff$ can be expressed as Taylor series:
\begin{equation} \label{S-bound}
S = \sum_{j=1}^\infty S_j
\end{equation}
\begin{equation} \label{H-bound}
\Heff = \sum_{j=1}^\infty \Heff^{(j)}.
\end{equation}
A systematic method for calculating the Taylor coefficients is given in~\cite[Section~3.2]{bravyi_2011}.
We will only need the first two coefficients:
\begin{equation} \label{H1}
\Heff^{(0)} = H_0 \Pi_- \quad\mathrm{and}\quad  \Heff^{(1)} = \Pi_- H_1 \Pi_-
\end{equation}

The size of the operators in the Taylor expansion can be bounded (see~\cite[Lemma~3.4]{bravyi_2011}):
\begin{equation}
\|S_j \| \le \BigO\left(\left(1+\frac{\lambda_0}{\pi\Delta} \right)\left(\frac{\|H_1\|}{\Delta} \right)^j\right)
\end{equation}

\begin{equation}
\|\Heff^{(j)} \| \le \BigO\left(\Delta \left(1+\frac{\lambda_0}{\pi\Delta} \right)\left(\frac{\|H_1\|}{\Delta} \right)^j\right)
\end{equation}
This implies~\cite{BH17}
\begin{equation}
\|S \| \le \BigO\left( \Delta^{-1}\| H_1\| \left(1+\frac{\lambda_0}{\pi\Delta} \right) \right)
\end{equation}
\begin{equation} \label{H-bound-error}
\|\Heff - \Heff(k) \| \le \BigO\left( \Delta^{-k}\| H_1\|^{k+1} \left(1+\frac{\lambda_0}{\pi\Delta} \right) \right)
\end{equation}
where $\Heff(k) = \sum_{j=1}^k \Heff^{(j)}$.

\section*{Acknowledgements}
J.\,B.~acknowledges support from the Draper's Junior Research Fellowship at Pembroke College.
T.\,S.\,C.~is supported by the Royal Society.
T.\,K.~is supported by the EPSRC Centre for Doctoral Training in Delivering Quantum Technologies [EP/L015242/1].
This work was supported by the EPSRC Prosperity Partnership in Quantum Software for Simulation and Modelling (EP/S005021/1).

\bibliography{ref-gc-quantum}
\end{document}